\documentclass[a4paper,11pt]{article}

\usepackage{graphicx}

\usepackage{fullpage}
\usepackage{libertine}
\usepackage{color}
\usepackage{subcaption}
\usepackage{multirow}

\usepackage[ruled,vlined]{algorithm2e}
\SetArgSty{textrm}

\usepackage{amsmath,amsfonts,amssymb,amsthm}

\usepackage[breaklinks=true]{hyperref}
\usepackage[svgnames]{xcolor}
\usepackage[capitalise,nameinlink]{cleveref}
\hypersetup{colorlinks={true},linkcolor={DarkBlue},citecolor=[named]{DarkGreen}}

\usepackage{natbib}
\usepackage{authblk}

\usepackage{subcaption}
\usepackage{tikz}  
\usetikzlibrary{arrows}
\usetikzlibrary{patterns,snakes}
\usetikzlibrary{decorations.shapes}
\tikzstyle{overbrace text style}=[font=\tiny, above, pos=.5, yshift=5pt]
\tikzstyle{overbrace style}=[decorate,decoration={brace,raise=5pt,amplitude=3pt}]
\usetikzlibrary{shapes.geometric}
\usepackage{pgfplots}
\pgfplotsset{compat=1.16}
\newtheorem{theorem}{Theorem}[section]

\newtheorem{lemma}[theorem]{Lemma}

\theoremstyle{definition}
\newtheorem*{comment*}{Comment}

\newcommand{\cost}{\text{cost}}
\newcommand{\AoM}{\text{Avg-of-Max}}
\newcommand{\MoA}{\text{Max-of-Avg}}

\newcommand{\favorite}{\text{top}}
\usepackage{mathtools}
\DeclarePairedDelimiter\abs{\lvert}{\rvert}
\newcommand{\OPT}{o}
\DeclareMathOperator*{\argmax}{arg\,max}
\newcommand{\info}{\text{info}}

\title{\bf Metric Distortion under Group-Fair Objectives}
\author[1]{Georgios Amanatidis}
\author[2]{Elliot Anshelevich}
\author[2]{\\Christopher Jerrett}
\author[1]{Alexandros A. Voudouris}

\affil[1]{University of Essex}
\affil[2]{Rensselaer Polytechnic Institute}

\date{}

\begin{document}

\allowdisplaybreaks

\maketitle

\begin{abstract}
We consider a voting problem in which a set of agents have metric preferences over a set of alternatives, and are also partitioned into disjoint groups. Given information about the preferences of the agents and their groups, our goal is to decide an alternative to approximately minimize an objective function that takes the groups of agents into account. We consider two natural group-fair objectives known as {\em Max-of-Avg} and {\em Avg-of-Max} which are different combinations of the max and the average cost in and out of the groups. We show tight bounds on the best possible {\em distortion} that can be achieved by various classes of mechanisms depending on the amount of information they have access to. In particular, we consider {\em group-oblivious full-information} mechanisms that do not know the groups but have access to the exact distances between agents and alternatives in the metric space, {\em group-oblivious ordinal-information} mechanisms that again do not know the groups but are given the ordinal preferences of the agents, and {\em group-aware} mechanisms that have full knowledge of the structure of the agent groups and also ordinal information about the metric space. 
\end{abstract}

\section{Introduction} \label{sec:intro}

One of the main subjects of study in (computational) social choice theory is to identify the capabilities and limitations of making appropriate collective decisions when given the preferences of individuals (or, {\em agents}) over alternative outcomes. This is done either by an axiomatic analysis of the potential decision-making mechanisms (which are also referred to as voting rules) \citep{comsoc-book}, or a qualitative analysis that aims to quantify the possible loss of efficiency when the agents have private cardinal utilities or costs for the alternatives but are only able to communicate partial information about their preferences, for example using ordinal information. This inefficiency is quantified by the notion of {\em distortion} which compares the quality of the computed outcome to that of the ideal outcome that could have been computed if full information about the underlying utilities of the agents was available. Since its introduction more than 15 year ago, distortion has been studied for many different social choice problems (such as voting applications, resource allocation, and facility location) and under different restrictions about the cardinal preferences of the agents (such as assuming unit-sum utilities or metric costs). For a more detailed overview see our discussion of the related work below and the survey of \cite{distortion-survey}.

With few exceptions, the distortion literature has focused on voting settings in which the agents are assumed to be independent of each other. As such, the quality of the different outcomes is typically calculated using measures such as the {\em social welfare} (defined as the total or average utility of all agents) or the {\em egalitarian welfare} (defined as the minimum utility over all agents). However, there are social choice applications in which, while the agents can act autonomously, they are also part of larger groups and care about the overall welfare of the members of their groups, but not that much about other groups. As a toy example, consider the case of a university department, the academics of which are members of different research groups. For several matters, such as electing the head of the department, each academic participates individually in the voting process, but the outcome might affect different groups in different ways. Due to this, objectives such as the social and the egalitarian welfare do not fully capture the quality of an outcome according to the structure of the problem. Instead, we would like objectives that take into account the partition of the agents into groups to measure efficiency and also satisfy other desired properties such as fairness or some form of balance among different groups. 

Two such natural objectives were first introduced by \citet{AFV22} who studied a metric district-based single-winner voting setting, where the agents have {\em costs} for the alternatives that are determined by their {\em distances} in a metric space, and the agents are furthermore partitioned into groups that are called {\em districts}. The first objective is {\em Max-of-Avg}, defined as the maximum over all groups of the average total distance of the agents within each group from the chosen alternative, and the second one is {\em Avg-of-Max}, defined as the average over all groups of the maximum distance among any agent within each group from the chosen alternative.\footnote{Observe that both of these objectives are essentially combinations of the {\em social cost} and the {\em egalitarian cost}, which are the analogues of the social welfare and egalitarian welfare when the agents have costs for the alternatives rather than utilities.}
By their definition, to optimize them, we need to choose outcomes that strike a balance between the average or maximum cost of any group as a whole, thus achieving fairness among different groups, on top of absolute efficiency.  

\subsection{Our Contribution}
We study a single-winner voting setting with $n$ agents and $m$ alternatives that lie in a metric space. Furthermore, the agents are partitioned into $k$ disjoint groups. Given some information about the groups of agents, as well as the distances between agents and alternatives in the metric space, our goal is to choose an alternative as the winner that is (approximately) efficient with respect to the Max-of-Avg and Avg-of-Max objectives that were defined above. In particular, we show tight bounds on the distortion of different classes of {\em deterministic} mechanisms, depending on the type of information they have access to in order to decide the winner. 

We start by considering {\em group-oblivious} mechanisms which have no knowledge of the groups of agents.
In Section~\ref{sec:full}, we consider the class of {\em full-information} group-oblivious mechanisms which have complete information about the distances between agents and alternatives in the metric space. For such mechanisms, we show a tight bound of $3$ for Max-of-Avg, a tight bound of $3$ for Avg-of-Max on instances in which the groups are symmetric (i.e., all groups have the same size), and a tight bound of $k$ for Avg-of-Max on general instances. 
In Section~\ref{sec:ordinal}, we turn our attention to {\em ordinal-information} group-oblivious mechanisms which are given as input the ordinal preferences of the agents over the alternatives in the form of rankings from the smallest distance to the largest. We show a tight bound of $5$ for Max-of-Avg, a tight bound of $5$ for Avg-of-Max on instances with symmetric groups, and a tight bound of $2k+1$ for Avg-of-Max on general instances. An overview of our results for group-oblivious mechanisms is given in Table~\ref{tab:overview}.

\renewcommand{\arraystretch}{1.3}
\begin{table}[t]
\centering
\begin{tabular}{c|c|c|c}
                            &           & Full-information & Ordinal-information \\ \hline 
Max-of-Avg                  &           & $3$ (Theorems~\ref{thm:full:max-of-avg:lower}, \ref{thm:full:max-of-avg:upper:general}) & $5$ (Theorems~\ref{thm:ordinal:max-of-avg:lower}, \ref{thm:ordinal:max-of-avg:upper}) \\ \hline
\multirow{2}{*}{Avg-of-Max} & Symmetric & $3$ (Theorems~\ref{thm:full:avg-of-max:lower:symmetric}, \ref{thm:full:avg-of-max:upper:symmetric}) & $5$ (Theorems~\ref{thm:ordinal:avg-of-max:lower:symmetric}, \ref{thm:ordinal:avg-of-max:upper:symmetric})\\
                            & Asymmetric & $k$ (Theorems~\ref{thm:full:avg-of-max:lower}, \ref{thm:full:avg-of-max:upper}) & $2k+1$ (Theorems~\ref{thm:ordinal:avg-of-max:lower}, \ref{thm:ordinal:avg-of-max:upper})\\  
\hline
\end{tabular}
\caption{An overview of our tight distortion bounds for the class of group-oblivious mechanisms.}
\label{tab:overview}
\end{table}

In Section~\ref{sec:aware} we turn our attention to {\em group-aware} mechanisms which know the structure of the groups of agents. Having full information about the metric space on top of this knowledge about the groups makes the problem of optimizing the Max-of-Avg and the Avg-of-Max objectives trivial by simply calculating the cost of every alternative. Consequently, we consider group-aware mechanisms that have access to the ordinal preferences of the agents over the alternatives instead. For instances with two alternatives, we prove a tight bound of $3$ on the distortion of such mechanisms for both objectives. For general instances, we show that the distortion is still $3$ when we are allowed to exploit more information about the metric space for the upper bound. In particular, we assume access to the distances between the alternatives. Resolving the distortion of group-aware mechanisms is probably the most challenging open question that our work leaves open; we discuss this in Section~\ref{sec:conclusion}.

\subsection{Related Work}
Inspired by worst-case analysis, \citet{procaccia2006distortion} introduced the distortion as a means of quantifying the inefficiency of voting mechanisms that base their decisions on the ordinal preferences of the agents over the alternative outcomes. Without restricting the possible underlying cardinal utilities of the agents, the distortion can be shown to be unbounded for most natural mechanisms. This led to subsequent works to study voting settings where it is assumed that the agents have underlying normalized utilities \citep{boutilier2015optimal,caragiannis2017subset,ebadian2022optimized,ebadian2023explainable}, or costs determined by distances in an unknown metric space~\citep{anshelevich2018approximating,gkatzelis2020resolving,kempe2022veto,charikar2022randomized,charikar24breaking,CSV22,jaworski2020committees}, or combinations of the two~\citep{gkatzelis2023both}. The distortion has also been studied for other social choice problems, such as participatory budgeting~\citep{benade2021participatory}, matching~\citep{filos2014RP,amanatidis2022matching}, as well as clustering~\citep{anshelevich2016blind,burkhardt2024low} and other graph problems where only ordinal information is available~\citep{abramowitz2017utilitarians}.
We refer to the survey of \citet{distortion-survey} for a more detailed exposition of the distortion framework and the problems it has been applied to. 

While the bulk of the distortion literature has focused on settings where ordinal or even less than ordinal information is available about the preferences of the agents, there has been recent interest in settings where it is also possible to elicit some cardinal information.
For example, the agents might be able to communicate a number of bits about their preferences~\citep{mandal2019efficient,mandal2020optimal,kempe2020communication}, or answer value queries related to their utilities about the alternatives~\citep{amanatidis2021peeking,amanatidis2022matching,Amanatidis2024dice,ma2021matching,caragiannis2023impartial,burkhardt2024low}, or provide more information in the form of intensities~\citep{abramowitz2019passion,kahng2023intesities} or threshold approvals~\citep{bhaskar2018truthful,benade2021participatory,ebadian2023approval,anshelevich2024approvals,latifian2024approval}. 
In our work, we also consider more than ordinal information in the case of full-information group-oblivious mechanisms, where the main source of inefficiency comes from not knowing the structure of the groups of agents. 

As already previously mentioned, the particular objective functions (Max-of-Avg and Avg-of-Max) that we consider in this paper have been studied in the context of distortion by \citet{AFV22} for single-winner distributed metric voting, and subsequently by \citet{voudouris2023tight} for the same setting, and by \citet{filos2024distributedFL} for distributed facility location on the line. In those settings, similarly to our model here, the agents are partitioned into groups that are called districts, and a mechanism works in two steps: First, for each district, it decides a representative alternative or location based on given information about the preferences of the agents in the district, and then it decides a winner or a facility location  based on information about the district representatives. Such distributed mechanisms can be thought of as members of the class of group-aware mechanisms in our setting when the groups are assumed to be known. However, they are very restricted as they essentially forget any detailed in-group information in the second step and instead rely only on the group representatives to make final decisions. The Max-of-Avg and Avg-of-Max objectives have also been considered in the context of mechanism design without money for altruistic facility location problems by  \citet{zhou2022sp-group,zhou2024altruism}.

\section{Preliminaries} \label{sec:prelim}
An instance $I$ of our voting problem consists of a set $N$ of $n \geq 2$ {\em agents} and a set $A$ of $m \geq 2$ {\em alternatives}. 
Agents and alternatives are represented by points in a metric space. 
We denote by $d(x,y)$ the {\em distance} between any two points $x$ and $y$ in the metric space; the distance function satisfies the properties $d(x,x) = 0$, $d(x,y) = d(y,x)$, and the triangle inequality $d(x,y) \leq d(x,z) + d(z,y)$ for any $x, y, z \in N \cup A$. 
The agents are partitioned into $k \geq 2$ pairwise disjoint {\em groups} which may be known or unknown; Let $G \vcentcolon= \{g_1, \ldots, g_k\}$ be the set of groups, and denote by $n_g$ the {\em size} of any group $g \in G$. If the groups are {\em symmetric}, to simplify our notation we write $n_g = \lambda = n/k$. 

A {\em mechanism} $M$ takes as input some information $\info(I)$ related to the groups of agents and the distances between agents and alternatives in the metric space. Based on this information, it outputs one of the alternatives as the {\em winner}, denoted by $M(\info(I))$. When the groups are assumed to be unknown (Sections~\ref{sec:full} and~\ref{sec:ordinal}), we consider two different classes of {\em group-oblivious} mechanisms depending on the type of information related to the metric space they have access to:
\begin{itemize}
    \item {\em Full-information} group-oblivious mechanisms have complete knowledge of the metric space, that is, they have access to the distances between all agents and alternatives.
    \item {\em Ordinal-information} group-oblivious mechanisms have access to the ordinal preferences of the agents over the alternatives according to their distances; that is, if $d(i,x) < d(i,y)$ for an agent $i$ and alternatives $x$ and $y$, then $i$ ranks $x$ higher $y$. 
\end{itemize} 
When the groups are assumed to be known (Section~\ref{sec:aware}), we consider {\em group-aware} mechanisms that have access to the ordinal preferences of the agents and---potentially---information related to the distances between alternatives (but not between agents, or between agents and alternatives). 

We are interested in designing socially efficient mechanisms according to collective cost objective functions that take the groups of the agents into account. In particular, we focus on the following two objectives: 
\begin{itemize}
    \item The {\em Max-of-Avg} cost of an alternative $x$ in a given instance $I$ is the maximum over all groups of the average total distance of the agents within each group from $x$, that is, 
    \begin{align*}
    \MoA(x\,|\,I) = \max_{g\in G} \bigg( \frac{1}{n_g} \sum_{i \in g} d(i,x) \bigg).
    \end{align*}
    \item The {\em Avg-of-Max} cost of an alternative $x$ in a given instance $I$ is the average over all groups of the maximum distance of any agent within each group from $x$, that is, 
    \begin{align*}
    \AoM(x\,|\,I) = \frac{1}{k}\sum_{g\in G} \max_{i \in g} d(i,x). 
    \end{align*}
\end{itemize}
Whenever the cost objective is clear from context, we will simplify our notation and write $\cost(x\,|\,I)$ for the cost of alternative $x$ in a given instance $I$. We will simplify our notation even more and write $\cost(x)$ when the instance is also clear from context.

Since the mechanisms we consider only have partial information about the groups of agents or the metric space, they cannot always identify the alternatives that optimize cost objectives which depend on the structure of the groups, like Max-of-Avg and Avg-of-Max. The loss of efficiency of a mechanism $M$ is captured by its {\em distortion}, which is the worst-case ratio (over all possible instances) of the cost of the alternative chosen by $M$ over the minimum possible cost of \textit{any} alternative, that is
\begin{align*}
    \sup_{I} \frac{\cost(M(\info(I))\,|\,I)}{\min_x \cost(x\,|\,I)}.
\end{align*}
Observe that the distortion of any mechanism is always at least $1$; we define $0/0=1$ for consistency. We aim to reveal the best possible distortion that can be achieved by mechanisms in this group voting setting. 

\section{Full-Information Group-Oblivious Mechanisms}\label{sec:full}
We start the presentation of our technical results with the class of full-information group-oblivious mechanisms; recall that such mechanisms have complete access to the distances between all agents and alternatives, which means that their inefficiency is solely due to being oblivious to how the agents are partitioned into groups. For the Max-of-Avg objective, we show a tight bound of $3$ on the distortion of full-information mechanisms (Section~\ref{sec:full:max-of-avg}). For the Avg-of-Max objective, we first show a tight bound of $3$ for instances in which the groups are symmetric, and a tight bound of $k$ for general instances with asymmetric groups (Section~\ref{sec:full:avg-of-max}). 

\subsection{Max-of-Avg} \label{sec:full:max-of-avg}

We begin by showing a lower bound of $3$ on the distortion of full-information group-oblivious mechanisms for the Max-of-Avg objective using an instance with symmetric groups. 

\begin{theorem} \label{thm:full:max-of-avg:lower}
For Max-of-Avg, the distortion of any full-information group-oblivious mechanism is at least $3-\varepsilon$ for any $\varepsilon > 0$, even when there are only two alternatives and the groups are symmetric.
\end{theorem}

\begin{proof}
Let $\varepsilon>0$ be any constant and $\lambda\in \mathbb{N}$ be such that $\lambda > \frac{6}{\varepsilon} - 2$. 
Consider the following instance with $n = \lambda(\lambda+1)$ agents and two alternatives with known locations on the line of real numbers:
\begin{itemize}
     \item Alternative $a$ is at $1$ and alternative $b$ is at $3$;
     \item There are $\lambda$ agents at $0$, $\lambda(\lambda-1)$ agents at $2$, and $\lambda$ agents at $4$.
\end{itemize}
Due to the symmetric locations of the alternatives and the agents, any of the two alternatives can be chosen as the winner. We assume the winner is $a$, without loss of generality. The agents might be partitioned into the following $k=\lambda+1$ symmetric groups of size $\lambda$ each:
\begin{itemize}
     \item The first group consists of all the $\lambda$ agents at $4$;
     \item Each of the remaining $\lambda$ groups consists of one agent at $0$ and $\lambda-1$ agents at $2$. 
\end{itemize}
The total distance of the agents in the first group is $3\lambda$ from $a$ and $\lambda$ from $b$, whereas the total distance of the agents in each of the remaining groups is $\lambda$ from $a$ and $\lambda+2$ from $b$.
Hence, $\cost(a) = 3$ and $\cost(b) = 1+\frac{2}{\lambda}$, leading to a distortion of at least $\frac{3\lambda}{\lambda+2} = 3 - \frac{6}{\lambda+2} > 3 - \varepsilon$.
\end{proof}

It is not hard to obtain a matching upper bound of $3$ by using a mechanism that chooses the winner to be any alternative that minimizes the total distance of all agents. In Appendix~\ref{app:full:max-of-avg} we present a refined analysis of this mechanism, by characterizing the worst-case distortion instances, and we obtain a distortion upper bound of $3-\frac{2\mu}{n}$, where $\mu$ is the smallest group size and $n$ is the number of agents.

\begin{theorem}\label{thm:full:max-of-avg:upper:general}
For Max-of-Avg, the distortion of a mechanism that returns an alternative who minimizes the total distance from all agents is at most $3$.
\end{theorem}

\begin{proof}
Let $w$ be an alternative that minimizes the total distance from all agents, and let $o$ be an optimal alternative (that minimizes the Max-of-Avg cost according to the unknown groups of the agents). By the definition of $w$, there must exist some group $\gamma$ such that
$\sum_{i\in \gamma} d(i,w) \leq \sum_{i \in \gamma} d(i,o)$; otherwise, the total distance of $o$ from all agents would be strictly less than that of $w$, thus contradicting the choice of $w$. By the definition of the objective function, we also have that $\cost(o) \geq \frac{1}{n_g}\sum_{i \in g} d(i,o)$ for every group $g$. Denoting by $g_w$ the group that determines the cost of $w$ and using the triangle inequality, we have
\begin{align*}
\cost(w) = \frac{1}{n_{g_w}} \sum_{i \in g_w} d(i,w) \leq \frac{1}{n_{g_w}} \sum_{i \in g_w} ( d(i,o) + d(w,o) ) \leq \cost(o) + d(w,o).
\end{align*}
Using the triangle inequality and the property of group $\gamma$ mentioned above, we further have that
\begin{align*}
d(w,o) = \frac{1}{n_\gamma} \sum_{i \in \gamma} d(w,o) \leq \frac{1}{n_\gamma} \sum_{i \in \gamma} ( d(i,w) + d(i,o) ) 
\leq 2 \cdot \frac{1}{n_\gamma} \sum_{i \in \gamma} d(i,o) &\leq 2\cdot \cost(o).
\end{align*}
Combining the two, we obtain $\cost(w) \leq 3\cdot \cost(o)$, i.e., the desired upper bound of $3$. 
\end{proof}

\subsection{Avg-of-Max} \label{sec:full:avg-of-max}

For the Avg-of-Max objective, we first focus on instances where the groups are symmetric (that is, every group consists of the same number $\lambda = n/k$ of agents) and show a tight bound of $3$. 

\begin{theorem} \label{thm:full:avg-of-max:lower:symmetric}
For Avg-of-Max, the distortion of any full-information group-oblivious mechanism is at least $3-\varepsilon$ for any $\varepsilon > 0$, even when there are two alternatives and the groups are symmetric.
\end{theorem}

\begin{proof}
Let $\varepsilon>0$ be any constant and $\lambda\in \mathbb{N}$ be such that $\lambda > \frac{8}{\varepsilon} - 3$.
We consider the same instance construction as in the proof of Theorem \ref{thm:full:max-of-avg:lower} on the line of real numbers. Recall that: 
\begin{itemize}
     \item Alternative $a$ is at $1$ and alternative $b$ is at $3$;
     \item There are $\lambda$ agents at $0$, $\lambda(\lambda-1)$ agents at $2$, and $\lambda$ agents at $4$.
\end{itemize}
We assumed that the winner is $a$, which is without loss of generality due to symmetry. The agents are partitioned into the $k=\lambda+1$ symmetric groups:
\begin{itemize}
     \item The first group consists of all the $\lambda$ agents at $4$;
     \item Each of the remaining $\lambda$ groups consists of one agent at $0$ and $\lambda-1$ agents at $2$. 
\end{itemize}
Therefore, $\cost(a) = \left(3+\lambda\right) / (\lambda + 1)$ and $\cost(b) =  \left(1 + 3\lambda \right) / (\lambda + 1)$, leading to a distortion of at least $\frac{3\lambda+1}{\lambda+3} =  3-\frac{8}{\lambda + 3} > 3 - \varepsilon$.
\end{proof}

The tight upper bound follows again by choosing any alternative that minimizes the total distance from all agents; hence, this very simple mechanism is best possible in terms of both the Max-of-Avg objective for general instances and the Avg-of-Max objective for instances with symmetric groups. 

\begin{theorem} \label{thm:full:avg-of-max:upper:symmetric}
For Avg-of-Max and symmetric groups, the distortion of a mechanism that returns an alternative who minimizes the total distance from all agents is at most $3$. 
\end{theorem}

\begin{proof}
Let $w$ be an alternative that minimizes the total distance from all agents, and denote by $o$ an optimal alternative (that minimizes the Avg-of-Max cost according to the $k$ unknown groups of agents).
Let $S_1, \ldots, S_\lambda$ be any partition of the agents into $\lambda = n/k$ disjoint sets of size $k$ such that each set consists of one agent per group; note that there are multiple such partitions of the agents since the groups are symmetric. By the definition of $w$, there must exist some $\ell \in [\lambda]$ such that $\sum_{i \in S_\ell}d(i,w) \leq \sum_{i \in S_\ell} d(i,o)$ since, otherwise, the total distance of $o$ from the agents would be strictly less than that of $w$, thus contradicting the choice of $w$. Let $i_g$ be a most-distant agent in group $g$ from $w$, i.e., $i_g\in \argmax_{i\in g} d(i, w)$. By matching each $i_g$ to a unique agent $f(i_g) \in S_\ell$ (i.e., $f: \{i_{g_1}, \ldots, i_{g_k}\} \to S_\ell$ is a bijection), we can rewrite the property of set $S_\ell$ as 
\[\sum_{g\in G} d(f(i_g),w) \leq \sum_{g\in G} d(f(i_g),o).\]  
In addition, by the definition of the objective function, we have that 
\[\cost(o) \geq \frac{1}{k} \sum_{g\in G} d(i_g,o)\] 
and 
\[\cost(o) \geq \frac{1}{k} \sum_{i \in S_\ell} d(i,o) = \frac{1}{k} \sum_{g\in G} d(f(i_g),o).\] 
Hence, by applying the triangle inequality twice, we obtain
\begin{align*}
\cost(w) 
&= \frac{1}{k} \sum_{g\in G} d(i_g,w) \\
&\leq \frac{1}{k} \sum_{g\in G} \big( d(i_g,o) + d(f(i_g),o) + d(f(i_g),w) \big) \\
& \leq 3 \cdot \cost(o),
\end{align*}
which shows the desired upper bound of $3$. 
\end{proof}


We now turn our attention to the general case where the groups might be asymmetric and show a tight bound of $k$.

\begin{theorem} \label{thm:full:avg-of-max:lower}
For Avg-of-Max, the distortion of any full-information group-oblivious mechanism is at least $k$, even when there are two alternatives.
\end{theorem}

\begin{proof}
Consider the following instance with $n=2k$ agents and two alternatives on the line of real numbers:
\begin{itemize}
    \item Alternative $a$ is at $0$ and alternative $b$ is at $1$;
    \item There are $k$ agents at $0$ and $k$ agents at $1$.
\end{itemize}
Due to symmetry, given only this information, any of the two alternatives can be chosen as the winner. Without loss of generality, we assume the winner is $a$. 
In that case, however, the agents might be split into $k$ groups as follows:
\begin{itemize}
    \item The first group consists of all agents at $0$ and one agent at $1$;
    \item Each of the remaining $k-1$ groups consists of a single agent at $1$.
\end{itemize}
Hence, $\cost(a) = 1$ and $\cost(b) = 1 / k$, leading to a distortion of $k$. 
\end{proof}

For the upper bound, we first remark that choosing any alternative that minimizes the total distance from all agents (as we did in the case of Avg-of-Max, or Max-of-Avg with symmetric groups) leads to a distortion of at least $2k+1$. Nevertheless, we can achieve a matching bound of $k$ by choosing any alternative that minimizes the maximum distance from the agents. 

\begin{theorem} \label{thm:full:avg-of-max:upper}
For Avg-of-Max, the distortion of a mechanism that returns an alternative who minimizes the maximum distance from any agent is at most $k$.    
\end{theorem}

\begin{proof}
Let $w$ be the chosen alternative and $o$ an optimal alternative. Let $i_w$ and $i_o$ be the most distant agents from $w$ and $o$, respectively. Then, by the definition of $w$, $d(i_w,w) \leq d(i_o,o)$. By the definition of $i_w$, $d(i,w) \leq d(i_w,w)$ for every agent $i$. 
Hence, 
\begin{align*}
    \cost(w) = \frac{1}{k} \sum_{g \in G} \max_{i \in g} d(i,w) \leq \frac{1}{k} \sum_{g \in G} d(i_w,w) = d(i_w,w).  
\end{align*}
On the other hand, 
\begin{align*}
    \cost(o) = \frac{1}{k} \sum_{g \in G} \max_{i \in g} d(i,o) \geq \frac{1}{k}\, d(i_o,o) \geq \frac{1}{k}\, d(i_w,w). 
\end{align*}
Consequently, the distortion is at most $k$. 
\end{proof}

\section{Ordinal-Information Group-Oblivious mechanisms}\label{sec:ordinal}
We now consider mechanisms that are given access to ordinal information about the preferences of the agents over the alternatives, but are still oblivious to how the agents are partitioned into groups. Recall that every agent $i$ reports a ranking of the alternatives such that,
if $d(i,x) < d(i,y)$ for alternatives $x$ and $y$, then $i$ ranks $x$ higher $y$.
For the Max-of-Avg objective, we show a tight bound of $5$ on the distortion of ordinal-information group-oblivious mechanisms. For the Avg-of-Max objective, we show that the distortion is exactly $5$ when the groups are symmetric, and is exactly $2k+1$ when the groups are asymmetric.

\subsection{Max-of-Avg} \label{sec:ordinal-max-of-avg}
We start by showing a lower bound of $5$ on the distortion of any mechanism. 

\begin{theorem} \label{thm:ordinal:max-of-avg:lower}
For Max-of-Avg, the distortion of any ordinal-information group-oblivious mechanism is at least $5 - \varepsilon$ for any $\varepsilon > 0$, even when there are only two alternatives and the groups are symmetric.
\end{theorem}

\begin{proof}
Let $\varepsilon>0$ be any constant and $\lambda \geq 3$ be an odd integer such that $\lambda > \frac{4}{\varepsilon} - 1$. 
Consider an instance with $n = \lambda^2+\lambda = \lambda(\lambda+1)$ agents and two alternatives $\{a,b\}$; clearly $\lambda^2+\lambda$ is an even number. Half of the agents prefer $a$ and the other half prefer $b$. With this information, any of the two alternatives can be chosen as the winner, so assume that the winner is $a$. The agents are partitioned into $k=\lambda+1$ symmetric groups of $\lambda$ agents each. 
Consider the scenario where the metric space is the line of real numbers and the grouping of the agents is as follows: 
\begin{itemize}
    \item $a$ is at $0$ and $b$ is at $2$;
    \item In the first group, all $\lambda$ agents prefer $b$ and are all positioned at $2+\frac{\lambda+1}{2\lambda}$.
    \item In each of the remaining $\lambda$ groups, there are $\frac{\lambda+1}{2}$ agents that prefer $a$ and are positioned at $1$, and $\frac{\lambda-1}{2}$ agents that prefer $b$ and are positioned at $2$.
\end{itemize}
The total distance of the agents in the first group is $\lambda\left(2+\frac{\lambda+1}{2\lambda}\right) = \frac{5\lambda+1}{2}$ from $a$ and $\lambda\cdot \frac{\lambda+1}{2\lambda} = \frac{\lambda+1}{2}$ from $b$. In each of the remaining $\lambda$ groups, the total distance of the agents therein is 
$\frac{\lambda+1}{2} + 2\cdot \frac{\lambda-1}{2} = \frac{3\lambda-1}{2}$ from $a$ and $\frac{\lambda+1}{2}$ from $b$. Consequently, $\cost(a)=\frac{5\lambda+1}{2\lambda}$ (realized by the first group) and $\cost(b)=\frac{\lambda+1}{2\lambda}$ (realized by any of the groups), leading to a distortion of at least $5 - \frac{4}{\lambda+1} > 5-\varepsilon$.
\end{proof}

We now show that there are ordinal-information group-oblivious mechanisms which do achieve this best possible bound of $5$. The {\em domination graph} of an alternative $x$ is a bipartite graph $G_x = (N,N,E_x)$ with the set of agents on both sides and set of (directed) edges such that $(i,j) \in E_x$ if and only if $i$ prefers $x$ over the most-preferred alternative $\favorite(j)$ of $j$, that is, $d(i,x) \leq d(i,\favorite(j))$. We focus on alternatives whose domination graphs attain perfect matchings. There are several voting rules that compute alternatives with this property, such as {\sc PluralityMatching}~\citep{gkatzelis2020resolving} and {\sc PluralityVeto}~\citep{kempe2022veto}. The distortion of these rules in terms of the social cost (the total distance of the agents) is known to be exactly $3$. We show the following property of such alternatives, which will be useful in some of our upper bounds. 

\begin{lemma}\label{lem:domination-3-distortion}
Given an instance, let $x$ be some alternative whose domination graph attains a perfect matching, and $y$ any other alternative. 
Then, 
\[d(x,y) \leq \frac{4}{n} \cdot \sum_{g \in G} \sum_{i \in g} d(i,y).\] 
\end{lemma}

\begin{proof}
Let $\boldsymbol{\mu}=(\mu(i))_i$ be the perfect matching in the domination graph $G_x$ of $x$; that is, agent $i$ is matched to agent $\mu(i)$. By the triangle inequality, the property of the domination graph that $d(i,x) \leq d(i,\favorite(\mu(i)))$, the fact that $M$ is a perfect matching, and the fact that $d(i,\favorite(i)) \leq d(i,y)$ for any $i$, we have
\begin{align*}
n \cdot d(x,y) 
= \sum_{i \in N} d(x,y) 
&\leq \sum_{i \in N} d(i,x) + \sum_{i \in N}  d(i,y) \\
&\leq \sum_{i \in N} d(i,\favorite(\mu(i))) + \sum_{i \in N}  d(i,y)\\
&\leq \sum_{i \in N} \big( d(i,y) + d(\mu(i),y) + d(\mu(i),\favorite(\mu(i))) \big) + \sum_{i \in N}  d(i,y) \\
&\leq \sum_{i \in N} d(i,y) + 2 \cdot \sum_{i \in N} d(\mu(i),y) + \sum_{i \in N}  d(i,y) \\
&= 4 \cdot \sum_{g \in G} \sum_{i \in g}  d(i,y).
\end{align*}
The statement now follows by dividing each side of the inequality by $n$. 
\end{proof}

We are now ready to show the upper bound of $5$ for Max-of-Avg.

\begin{theorem} \label{thm:ordinal:max-of-avg:upper}
For Max-of-Avg, the distortion of a mechanism that returns an alternative whose domination graph has a perfect matching is at most $5$.
\end{theorem}

\begin{proof}
Let $w$ be the chosen alternative (whose domination graph has a perfect matching), and $o$ an optimal alternative. Let $g_w$ be the group that determines the maximum cost of $w$. By the definition of Max-of-Avg, we have $n_g \cdot \cost(o) \geq  \sum_{i \in g} d(i,o)$ for any group $g$. Since $n = \sum_g n_g$, by adding all these inequalities together, we have 
\begin{align}
n \cdot \cost(o) \geq \sum_{g \in G} \sum_{i \in g} d(i,o). \label{eq:max-of-avg:cost-o}
\end{align}
By the triangle inequality, we have
\begin{align}
\cost(w) 
&= \frac{1}{n_{g_w}} \sum_{i \in g_w} d(i,w) \nonumber \\
&\leq \frac{1}{n_{g_w}} \sum_{i \in g_w} d(i,o) + \frac{1}{n_{g_w}} \sum_{i \in g_w} d(w,o) \nonumber \\
&\leq \cost(o) + d(w,o).  \label{eq:max-of-avg:cost-w-to-bound}
\end{align}
By Lemma~\ref{lem:domination-3-distortion} with $x=w$ and $y=o$, and using \eqref{eq:max-of-avg:cost-o}, we have
\begin{align*}
    d(w,o) \leq \frac{4}{n} \cdot \sum_{g \in G} \sum_{i \in g} d(i,o) \leq 4 \cdot\cost(o), 
\end{align*}
which, combined with \eqref{eq:max-of-avg:cost-w-to-bound}, leads to 
\[\cost(w) \leq 5 \cdot \cost(o),\]
which directly implies the desired upper bound.  
\end{proof}

\subsection{Avg-of-Max} \label{sec:ordinal-avg-of-max}

For the Avg-of-Max cost, we first consider the case of symmetric groups, in which $n_g = \lambda$ for every $g$, and show a tight bound of $5$ on the distortion of ordinal-information group-oblivious mechanisms.  

\begin{theorem} \label{thm:ordinal:avg-of-max:lower:symmetric}
For Avg-of-max, the distortion of any ordinal-information group-oblivious mechanism is at least $5$, even when there are only two alternatives and the groups are symmetric. 
\end{theorem}

\begin{proof}
Let $\varepsilon>0$ be any constant and $\lambda \geq 2$ be an integer such that $\lambda > \frac{4}{\varepsilon}$. 
Consider an instance with $n=2\lambda(\lambda-1)$ agents and two alternatives $\{a,b\}$; clearly, $n$ is even. Half of the agents prefer $a$ while the remaining half prefer $b$. With this information, any of the two alternatives can be chosen as the winner, so assume that the winner is $a$. The agents might be split into $k = 2(\lambda-1)$ groups of $\lambda$ agents each as follows:
\begin{itemize}
    \item There are $\lambda$ groups, each consisting of $\lambda-1$ agents that prefer $a$ and one agent that prefers $b$;
    \item There are $\lambda-2$ groups, each consisting of $\lambda$ agents that prefer $b$.
\end{itemize}
Further, consider the metric space being the line of real numbers and the positioning of the alternatives and the agents being as follows: 
\begin{itemize}
    \item $a$ is at $0$ and $b$ is at $2$;
    \item All agents that prefer $a$ are at $1-\varepsilon/10$;
    \item The $\lambda$ agents that prefer $b$ and are part of the first $\lambda$ groups (in which there are agents that prefer $a$) are at $3$;
    \item The remaining $\lambda(\lambda-2)$ agents that prefer $b$ are at $2$. 
\end{itemize}
We have that 
\begin{align*}
    k \cdot \cost(a) = \lambda \cdot 3 + (\lambda-2)\cdot 2 = 5\lambda-2
\end{align*}
and 
\begin{align*}
    k \cdot \cost(b) = \lambda \cdot (1 + \varepsilon/10) + (\lambda-2)\cdot 0 = \lambda(1 + \varepsilon/10),
\end{align*}
leading to a distortion of at least $\frac{5}{1 + \varepsilon/10} - \frac{2}{\lambda(1 + \varepsilon/10)}  > 5- \frac{\varepsilon}{2} - \frac{2}{\lambda}> 5- \varepsilon$, where the first inequality is just a matter of simple calculations.
\end{proof}

For the upper bound, we consider again mechanisms that output alternatives whose domination graphs have perfect matchings, and show an upper bound of $5$ with a proof similar to the one used for the Max-of-Avg objective. 

\begin{theorem} \label{thm:ordinal:avg-of-max:upper:symmetric}
For Avg-of-Max and symmetric groups, the distortion of a mechanism that returns an alternative whose domination graph has a perfect matching is at most $5$.
\end{theorem}

\begin{proof}
Consider any instance with $k$ symmetric groups, each consisting of $\lambda = n/k$ agents. 
Let $w$ be an alternative whose domination graph has a perfect matching, and $o$ an optimal alternative. 
For every group $g$, let $i_g$ and $i_g^*$ be most-distant agents from $w$ and $o$, respectively. 
Clearly, 
\[\cost(o) = \frac{1}{k}\cdot\sum_{g \in G} d(i_g^*,o) \geq \frac{1}{k}\cdot \sum_{g \in G} d(i_g,o).\]
By the triangle inequality, we have
\begin{align}
\cost(w) 
&= \frac{1}{k}\cdot \sum_{g \in G} d(i_g,w) \nonumber \\
&\leq \frac{1}{k}\cdot \sum_{g \in G} d(i_g,o) + \frac{1}{k}\cdot \sum_{g \in G} d(w,o) \nonumber \\
&\leq \cost(o) + d(w,o)  \label{eq:avg-of-max:symmetric:cost-w-to-bound}
\end{align}
By Lemma~\ref{lem:domination-3-distortion} with $x=w$ and $y=o$, and since $k = n/\lambda$, we have
\begin{align*}
d(w,o) &\leq \frac{4}{n} \cdot \sum_{g \in G} \sum_{i \in g} d(i,o) \\
&\leq \frac{4}{n} \cdot \sum_{g \in G} \lambda \cdot \max_{i \in g} d(i,o) \\
&= \frac{4}{k} \sum_{g \in G} d(i_g^*,o) \\
&= 4 \cdot \cost(o)
\end{align*}
Using this, \eqref{eq:avg-of-max:symmetric:cost-w-to-bound} becomes
\begin{align*}
    \cost(w) \leq 5 \cdot \cost(o),
\end{align*}
giving us the desired bound of $5$ on the distortion. 
\end{proof}

For general instances with asymmetric groups, we show a tight bound of $2k+1$.

\begin{theorem} \label{thm:ordinal:avg-of-max:lower}
For Avg-of-Max, the distortion of any ordinal-information group-oblivious mechanism is at least $2k+1$, even when there are only two alternatives.
\end{theorem}

\begin{proof}
Consider the following instance with $n=2k$ agents and two alternatives located on the line of real numbers:
\begin{itemize}
    \item Alternative $a$ is at $0$ and alternative $b$ is at $2$;
    \item There $k$ agents that prefer alternative $a$ and $k$ agents that prefer alternative $b$.
\end{itemize} 
Since there is no way of distinguish between the two alternative given the preferences of the agents, we may assume that the winner is $a$, without loss of generality.
The agents might be partitioned into the following $k$ groups:
\begin{itemize}
    \item The first group consists of $k+1$ agents that includes those that prefer $a$ who are located at $1$ and one agent that prefers $b$ who is located at $3$; 
    \item Each of the remaining $k-1$ groups consist of just one agent that prefers $b$ who is located at $2$.
\end{itemize}
Hence, $k \cdot \cost(a) = 3 + (k-1)\cdot 2 = 2k+1$ and $k \cdot \cost(b) = 1$, leading to a distortion of $2k+1$. 
\end{proof}

The matching upper bound follows easily by choosing any alternative who is ranked first by some agent.

\begin{theorem} \label{thm:ordinal:avg-of-max:upper}
For Avg-of-Max, the distortion of a mechanism that returns an alternative who is the most-preferred of some agent is at most $2k+1$.
\end{theorem}

\begin{proof}
For any group $g$, let $i_g$ and $i_g^*$ be agents that are most-distant from the winner $w$ and the optimal alternative $o$, respectively. 
Let $S$ be the set of groups in which there is at least one agent with $w$ as her most-preferred alternative, and observe that $|S| \geq 1$, and thus $|G\setminus S| \leq k-1$. 
We make the following observations:
\begin{itemize}
\item 
For any group $g \in S$, let $j_g$ be an agent who ranks $w$ first. 
By the triangle inequality, for any $g \in S$, we have that
\begin{align*}
    d(i_g,w) \leq d(i_g,o) + d(j_g,o) + d(j_g,w) \leq d(i_g,o) + 2d(j_g,o) \leq 3 \cdot d(i_g^*,o). 
\end{align*}
In addition, since there is agent $j_g$ that prefers $w$ over $o$, then 
\[ \frac12 \cdot d(w,o) \leq \frac12 \big( d(j_g,w) + d(j_g,o) \big) \leq d(j_g,o) \leq d(i_g^*,o). \]

\item
For any group $g \not\in S$, by the triangle inequality, we have that
\begin{align*}
d(i_g,w) \leq d(i_g,o) + d(w,o) \leq d(i_g^*,o) + d(w,o).
\end{align*}
Also $d(i_g^*,o) \geq 0$.

\end{itemize}

Using these, we can now bound the distortion as follows:
\begin{align*}
\frac{\cost(w)}{\cost(o)}
&= \frac{\sum_{g \in S}d(i_g,w) + \sum_{g \not\in S}d(i_g,w)}{\sum_{g \in G} d(i_g^*,o)} \\
&\leq \frac{3\sum_{g \in S}d(i_g^*,o) + \sum_{g \not\in S}\big(d(i_g^*,o) + d(w,o) \big)}{\sum_{g \in G} d(i_g^*,o)} \\
&\leq 3 + \frac{\sum_{g \not\in S} d(w,o) }{\sum_{g \in S} d(i_g^*,o)} \\
&\leq 3 + \frac{|G\setminus S|\cdot d(w,o)}{|S| \cdot \frac{1}{2}\cdot d(w,o)} \\
&\leq 3 + 2(k-1) = 2k+1,
\end{align*}
as desired.
\end{proof}

\section{Group-Aware Mechanisms} \label{sec:aware}

In the previous two sections, we focused on mechanisms that are oblivious to the partition of the agents into groups. It is thus natural for one to wonder whether improved distortion bounds can be achieved by mechanisms that are {\em aware} of the groups. Clearly, we can optimize exactly both objectives if we are also given full information about the locations of the agents and the alternatives in the metric space, so this question makes sense when we only have access to partial information about the metric space, such as ordinal information. In this section, we consider such group-aware mechanisms and show tight bounds on the distortion in two cases: 
(1) there are only two alternatives; 
(2) there are $m \geq 2$ alternatives and the distances between them are known. 

\subsection{The Case of Two Alternatives}
Here, we consider the case of two alternative $a$ and $b$. For both objectives (Max-of-Avg and Avg-of-Max), we show a tight bound of $3$ on the distortion of ordinal-information mechanisms. We start with the lower bounds, which are implied by the classic voting setting without groups. 

\begin{theorem}
For both Max-of-Avg and Avg-of-Max, the distortion of any ordinal-information group-aware mechanism is at least $3$, even when there are only two alternatives and the groups are symmetric. 
\end{theorem}

\begin{proof}
The lower bounds for both objectives follow by considering instances in which the agents are partitioned into singleton groups. Then, the Max-of-Avg objective reduces to the egalitarian cost (the maximum distance over all agents), while the Avg-of-Max objectives reduces to the average social cost (the average total distance of the agents). When there are no groups (or, equivalently, there are singleton groups), the best possible distortion in terms of the egalitarian or the average social cost is $3$, even where there are only two alternatives~\citep{anshelevich2018approximating,gkatzelis2020resolving,kempe2022veto}.
\end{proof}

Next, we present the tight upper bounds. For the Max-of-Avg objective, we consider the {\sc Group-Proportional-Majority} mechanism which chooses the winner $w$ to be an alternative that has the largest proportional majority within any group. In particular, for any alternative $x\in \{a,b\}$, let $n_g(x)$ be the number of agents in group $g$ that prefer $x$. Then, 
\[w \in \argmax_{x \in \{a,b\}}  \max_{g \in G} \frac{n_g(x)}{n_g}.\]

\begin{theorem}\label{thm:aware:max-of-avg:two:upper}
For Max-of-Avg and two alternatives, the distortion of {\sc Group-Proportional-Majority} is at most $3$. 
\end{theorem}

\begin{proof}
For any group $g$, let $S_g(x)$ be the set of agents in $g$ that prefer $x$; thus, $n_g(x) = |S_g(x)|$.
By the definition of the mechanism, there is a group $\gamma$ such that $\frac{n_{\gamma}(w)}{n_{\gamma}} \geq \frac{n_g(o)}{n_g}$ for every group $g$. Clearly, for any agent $i \in S_{\gamma}(w)$, $d(i,w) \leq d(i,o)$, and thus, by the triangle inequality, $d(i,o) \geq d(w,o)/2$. Using this, for any group $g$, we can bound the optimal cost as follows:
\[\cost(o) \geq \frac{1}{n_{\gamma}} \sum_{i \in \gamma} d(i,o) 
\geq \frac{1}{n_{\gamma}} \sum_{i \in S_{\gamma}(w)} d(i,o) 
\geq \frac{n_{\gamma}(w)}{n_\gamma} \cdot \frac{d(w,o)}{2} 
\geq \frac{n_{g}(o)}{n_g} \cdot \frac{d(w,o)}{2}\]
or, equivalently, 
\begin{equation}
\frac{n_{g}(o)}{n_g} \cdot d(w,o) \leq 2 \cdot \cost(o). \label{eq:aware:max-of-avg:upper:optimal}
\end{equation}
Now, let $g_w$ be the group that determines the cost of $w$. Using the fact that $d(i,w) \leq d(i,o)$ for every agent $i \in S_{g_w}(w)$ and the triangle inequality, we have
\begin{align*}
\cost(w) 
&= \frac{1}{n_{g_w}}\sum_{i \in g_w} d(i,w) \\
&= \frac{1}{n_{g_w}} \sum_{i \in S_{g_w}(w)} d(i,w) + \frac{1}{n_{g_w}} \sum_{i \in S_{g_w}(o)} d(i,w) \\
&\leq \frac{1}{n_{g_w}} \sum_{i \in S_{g_w}(w)} d(i,o) + \frac{1}{n_{g_w}} \sum_{i \in S_{g_w}(o)} \!( d(i,o) + d(w,o) ) \\
&\leq \cost(o) + \frac{n_{g_w}(o)}{n_{g_w}} \cdot d(w,o).
\end{align*}
Using \eqref{eq:aware:max-of-avg:upper:optimal} for $g = g_w$, we finally obtain $\cost(w) \leq 3 \cdot \cost(o)$, as desired.  
\end{proof}

For Avg-of-Max, we consider the {\sc Group-Score} mechanism  which, for any alternative $x \in \{a,b\}$, assigns $2$ points to $x$ for any group in which all agents prefer $x$, and $1$ point for any group in which some agents prefer $x$ while the remaining agents prefer the other alternative. The winner $w$ is the alternative with maximum score, breaking possible ties arbitrarily. 

\begin{theorem} \label{thm:aware:avg-of-max:two:upper}
For Avg-of-Max and two alternatives, the distortion of {\sc Group-Score} is at most $3$. 
\end{theorem}

\begin{proof}
Let $w$ be the alternative chosen by the mechanism, and $o$ an optimal alternative; clearly, if $w=o$, the distortion is $1$, so we assume that $w \neq o$. We partition the groups into three sets: 
\begin{itemize}
    \item $S_w$ contains the groups that are {\em in favor of $w$}, in which all agents prefer $w$ over $o$;
    \item $S_o$ contains the groups that are {\em in favor of $o$}, in which all agents prefer $o$ over $w$;
    \item $S_m$ contains the groups that are {\em mixed}, in which some agents prefer $w$ over $o$ and some agents prefer $o$ over $w$. 
\end{itemize}
For any group $g$, let $i_g$ be a most-distant agent from $w$ and $i_g^*$ a most-distant agent from $o$; hence, $\cost(w) = \frac{1}{k} \sum_g d(i_g,w)$ and $\cost(o) = \frac{1}{k}\sum_g d(i_g^*,o)$. 
We make the following observations:
\begin{itemize}
    \item For any $g \in S_w$, both $i_g$ and $i_g^*$ prefer $w$ over $o$. Hence, $d(i_g,w) \leq d(i_g,o) \leq d(i_g^*,o)$ and, using the triangle inequality, $d(i_g^*,o) \geq \frac12 \cdot d(w,o)$.
    \item For any $g \in S_o$, by the triangle inequality, $d(i_g,w) \leq d(i_g,o) + d(w,o) \leq d(i_g^*,o) + d(w,o)$. Also, recall that $d(i_g^*,o) \geq 0$.
    \item For any $g \in S_m$, like above, $d(i_g,w) \leq d(i_g^*,o) + d(w,o)$. Also, since there is at least one agent that prefers $w$ over $o$, it must be the case that $d(i_g^*,o) \geq \frac12 \cdot d(w,o)$.
\end{itemize}
Using first the upper bounds on the distances from $w$, and then the lower bounds on the distances from $o$, we can write the distortion as follows:
\begin{align*}
\frac{\cost(w)}{\cost(o)} & = \frac{\sum_g d(i_g,w)}{\sum_g d(i_g^*,o)} \\
&\leq \frac{\sum_g d(i_g^*,o) + (|S_o| +|S_m|) \cdot d(w,o) }{\sum_g d(i_g^*,o)} \\
&= 1 + \frac{(|S_o| +|S_m|) \cdot d(w,o) }{\sum_g d(i_g^*,o)} \\
&\leq 1 + \frac{(|S_o| +|S_m|) \cdot d(w,o) }{(|S_w| +|S_m|)\cdot \frac12 \cdot d(w,o)} \\
&= 1 + 2 \cdot \frac{|S_o| +|S_m|}{|S_w| +|S_m|}.
\end{align*}
By the definition of the mechanism, $w$ is chosen as the winner because $2|S_w| + |S_m| \geq 2|S_o| + |S_m|$ or, equivalently, $|S_w| \geq |S_o|$. Using this, the distortion is at most
\begin{align*}
1 + 2 \cdot \frac{|S_o| +|S_m|}{|S_w| +|S_m|} 
\leq 1 + 2 \cdot \frac{|S_o| +|S_m|}{|S_o| +|S_m|} 
= 3,
\end{align*}
as claimed.
\end{proof}

\subsection{Known Distances between Alternatives}
We finally consider the general case of $m\geq 2$ but when slightly more information than just ordinal preferences is available. In particular, besides knowing the ordinal preferences of the agents over the alternatives, we assume that the distances between the alternatives in the metric space are also known. This is a natural assumption in various important applications (such as in facility location problems) and it has thus been examined in previous work on the distortion for different voting settings~\citep{anshelevich2021ordinal,anshelevich2024approvals}. Before we continue, we remark that the lower bound of $3$, and even the lower bounds in the previous sections, still hold for this setting where the distances between the alternatives are known since they have been proven using instances with just two alternatives.

To show a tight bound of $3$ for the two objectives, we consider mechanisms that virtually map each agent $i$ to its most-preferred alternative $\favorite(i)$, and then choose the winner to be an alternative that minimizes the objective under consideration for these most-preferred alternatives. In particular, the winner for the Max-of-Avg objective is
\begin{align*}
    w \in \argmax_{x \in A}  \max_{g \in G} \bigg( \frac{1}{n_g} \sum_{i \in g} d(\favorite(i),x) \bigg),
\end{align*}
while the winner for the Avg-of-Max objective is
\begin{align*}
    w \in \argmax_{x \in A} \bigg( \frac{1}{k} \sum_{g \in G} \max_{i \in g}  d(\favorite(i),x)\bigg).
\end{align*}
We will refer to these two mechanisms as {\sc Virtual-MiniMax-of-Avg} and {\sc Virtual-MiniAvg-of-Max}, respectively. 

\begin{theorem}\label{thm:aware:known:upper}
When the alternative locations are known, the distortion of\, {\sc Virtual-MiniMax-of-Avg} is at most $3$ for Max-of-Avg, and the distortion of\, {\sc Virtual-MiniAvg-of-Max} is at most $3$ for Avg-of-Max. 
\end{theorem}

\begin{proof}
We first show the bound for the Max-of-Avg objective. Let $w$ be the alternative chosen by the {\sc Virtual-MiniMax-of-Avg} mechanism, and denote by $o$ an optimal alternative. By definition, $\cost(o) \geq \max_{g \in G}\big( \frac{1}{n_g} \sum_{i \in g} d(i,o) \big)$. 
Let $g_w$ be the group that determines the cost of $w$. 
By the triangle inequality, the fact that $d(i,\favorite(i)) \leq d(i,o)$  for any agent $i$, the definition of $w$ (which minimizes the Max-of-Avg cost of the most-preferred alternatives of all agents), and the fact that the maximum of a set of additive functions is subadditive, we obtain
\begin{align*}
\cost(w) 
&= \frac{1}{n_{g_w}} \sum_{i \in g_w} d(i,w) \\
&\leq \frac{1}{n_{g_w}} \sum_{i \in g_w} d(i,\favorite(i)) + \frac{1}{n_{g_w}} \sum_{i \in g_w} d(\favorite(i), w) \\
&\leq \frac{1}{n_{g_w}} \sum_{i \in g_w} d(i,o) + \max_{g \in G} \bigg( \frac{1}{n_g} \sum_{i \in g} d(\favorite(i), w) \bigg) \\
&\leq \cost(o) + \max_{g \in G} \bigg( \frac{1}{n_g} \sum_{i \in g} d(\favorite(i), o) \bigg)  \\
&\leq \cost(o) + \max_{g \in G} \bigg( \frac{1}{n_g} \sum_{i \in g} d(i,\favorite(i)) \bigg) +  \max_{g \in G} \bigg( \frac{1}{n_g} \sum_{i \in g} d(i,o) \bigg) \\
&\leq 3\cdot \cost(o).
\end{align*}

The proof for the Max-of-Avg objective is quite similar. Now let $w$ be the alternative chosen by the {\sc Virtual-MiniAvg-of-Max} mechanism.
For the optimal alternative $o$, by definition, we have
$\cost(o) \geq \frac{1}{k} \sum_{g \in G} \max_{i \in g} d(i,o)$.
Let $i_g$ be the most-distant agent from $w$ in group $g$. 
Again, using the triangle inequality, the fact that $d(i,\favorite(i)) \leq d(i,o)$ for any agent $i$, the definition of $w$ (which now minimizes the Avg-of-Max cost of the most-preferred alternative of all agents), and the fact that max is a subadditive function, we obtain
\begin{align*}
\cost(w) 
&= \frac{1}{k} \sum_{g \in G} d(i_g,w) \\
&\leq \frac{1}{k} \sum_{g \in G} d(i_g,\favorite(i_g)) + \frac{1}{k} \sum_{g \in G} d(\favorite(i_g), w) \\
&\leq \frac{1}{k} \sum_{g \in G} d(i_g,o) + \frac{1}{k} \sum_{g \in G} \max_{i \in g} d(\favorite(i), w) \\
&\leq \frac{1}{k} \sum_{g \in G} \max_{i \in g} d(i,o)  + \frac{1}{k} \sum_{g \in G} \max_{i \in g} d(\favorite(i), o) \\
&\leq \cost(o) + \frac{1}{k} \sum_{g \in G} \max_{i \in g} \big(d(i,\favorite(i)) + d(i,o)\big)  \\
&\leq \cost(o) + \frac{1}{k} \sum_{g \in G} \max_{i \in g} \big(2 \cdot d(i,o)\big)  \\
&\leq 3\cdot \cost(o),
\end{align*}
as claimed.
\end{proof}

\section{Conclusion and Open problems} \label{sec:conclusion}
In this paper, we considered a metric voting setting in which the agents are partitioned into groups. When the groups are unknown, we showed tight bounds on the distortion of oblivious full-information and oblivious ordinal-information mechanisms in terms of two objectives that take the groups into account, the Max-of-Avg and the Avg-of-Max objectives. On the other hand, when the groups are known, we managed to show tight bounds on the distortion of group-aware ordinal mechanisms when there are just two alternatives or when we also have access to the locations of the alternatives in the metric space. 

There are multiple avenues for further research in the group voting model we considered here. The most important problem that our work leaves open is to resolve the distortion of group-aware ordinal mechanisms for more than two alternatives. While this a very challenging task in general, we remark that achieving constant distortion can be done by using the two-step distributed mechanisms of \citet{AFV22} which are, by definition, group-aware. However, those mechanisms do not fully exploit the structure of the groups, and we therefore expect that better distortion bounds can be achieved by unlocking the full potential of group-aware mechanisms.
Other interesting directions would be to consider randomized mechanisms and other objective functions that take the groups into account, beyond Max-of-Avg and Avg-of-Max.

\bibliographystyle{plainnat}
\bibliography{references}

\newpage
\appendix

\section{Refined Analysis of the Full-Information Group-Oblivious \\ Mechanism for Max-of-Avg} \label{app:full:max-of-avg}
In Theorem \ref{thm:full:max-of-avg:upper:general}, we showed that choosing any alternative that minimizes the total distance from all agents achieved the best possible distortion of $3$ in terms of the Max-of-Avg objective when taking the worst-case over all possible instances.  In this appendix, we present a more detailed analysis of this mechanism and show that the distortion bound is $3-2\mu/n$, where $\mu$ is the smallest group size and $n$ is the number of agents. While this is still $3$ in the worst-case, it implies some improved upper bounds for cases in which the number of groups $k$ is small or the smallest group size $\mu$ is rather large compared to $n$. In particular, for instances where there are $k$ symmetric groups, the bound becomes $3-2/k$. 

\begin{theorem} \label{thm:full:max-of-avg:upper:improved-analysis}
For Max-of-Avg, the distortion of any alternative that minimizes the total distance from all agents is at most $3-2\mu/n$, where $\mu$ is the smallest group size and $n$ is the number of agents.
\end{theorem}

\begin{proof}
Let $w$ be the alternative that minimizes the total distance from all agents. 
Suppose that $d(w,o) = 2$ without loss of generality. 
Let $g_w$ be the group that determines the cost of $w$ (that is, $g_w$ is the group for which the average distance of the agents therein from $w$ is maximized). 
We define $z_w$ and $z_o$ to be the average distance of all agents that do not belong to $g_w$ from $w$ and $o$, respectively, that is
\[z_w := \frac{1}{n-n_{g_w}} \sum_{g \neq g_w} \sum_{i \in g} d(i,w)\] 
and 
\[z_\OPT := \frac{1}{n-n_{g_w}}\sum_{g \neq g_w} \sum_{i \in g} d(i,\OPT).\]
Also, let $y = \frac{1}{n_{g_w}}\sum_{i \in g_w} d(i,\OPT)$ and $\gamma = \cost(w) = \frac{1}{n_{g_w}}\sum_{i \in g_w} d(i,w)$.
Using this notation, we have
\[\sum_{i} d(i,w) = n_{g_w}\cdot \gamma + (n - n_{g_w})\cdot z_w\] 
and
\[\sum_{i} d(i,\OPT) = n_{g_w} \cdot y + (n - n_{g_w})\cdot z_\OPT.\] 

By applying the triangle inequality, and using our assumption that $d(w,o)=2$, we obtain
\begin{align*}
\gamma \leq \frac{1}{n_{g_w}}\sum_{i \in g_w} ( d(i,o) + d(o, w) ) = 2+y.
\end{align*} 
We also have that $y \leq \gamma$; otherwise, since $\gamma = \cost(w)$ and $y \leq \cost(o)$, the distortion would be $1$. This implies that 
$\sum_{i \in g_w} d(i, w) \geq \sum_{i \in g_w} d(i, \OPT)$. Also, by the definition of $w$ (which minimizes the total distance from all agents), $\sum_i d(i, w) \leq \sum_i d(i, \OPT)$. By these, we can conclude that $z_w \leq z_\OPT$; otherwise the total distance from $o$ would be strictly smaller than $w$, and $w$ would not be the winner. So, there must exist $C, x \geq 0$ such that $z_w = C - x$ and $z_\OPT = C + x$. 
In fact, using the triangle inequality,
\begin{align*}
    2=d(w,o) = \frac{1}{n-n_{g_w}}\sum_{g \neq g_w} \sum_{i \in g} d(w,o) \leq z_w + z_\OPT = 2C
\end{align*}
and thus $C \geq 1$.

    \begin{figure}[t]
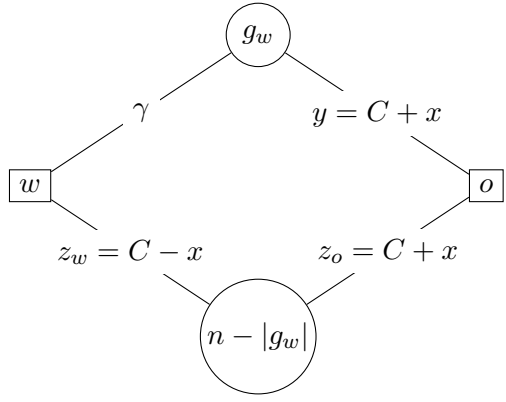

        \centering
        \tikz {
          \node (gw) [circle, draw] at (0,2) {$g_w$};
          \node (s) [circle, draw, inner sep=1.5pt] at (0,-2)  {$n-\abs{g_w}$};
          
          \node (opt) [rectangle, draw] at (-3,0) {$w$};
          \node (w) [rectangle, draw]  at (3,0) {$\OPT$};
          
          \draw (gw) edge node [midway, fill=white] {$y=C+x$} (w) ;
          \draw (gw) edge node [midway, fill=white] {$\gamma$} (opt) ;
          \draw (s) edge node [midway, fill=white] {$z_o = C+x$} (w) ;
          \draw (s) edge node [midway, fill=white] {$z_w = C-x$} (opt) ;
        }
        \caption{An illustration of the worst-case instance given by Lemma \ref{lem:reduction_lemma}.}
        \label{fig:reduction_figure_1}
    \end{figure}
    
We now introduce a lemma that characterizes the worst-case instances we need to focus on. 

\begin{lemma}\label{lem:reduction_lemma}
In a worst-case instance (in terms of distortion), $y = z_\OPT = C+x$.
\end{lemma}

\begin{proof}
The proof proceeds in two steps: 
We first transform the original instance $I$ into a new instance $I'$ in which there are only a few distinct points in the metric space where agents and alternatives are located. 
Then, we transform $I'$ into another instance $I''$ with the desired property $y = C+x$. 
While doing these transformations we will show that the cost of $\OPT$ does not increase, that is, 
\[\cost(\OPT\,|\,I) \geq \cost(\OPT\,|\,I') \geq \cost(\OPT\,|\,I''),\] 
while the cost of $w$ does not decrease, that is, 
\[\cost(w\,|\,I) \leq \cost(w\,|\,I') \leq \cost(w\,|\,I'').\]

We now introduce the first transformation. 
Consider the following new instance $I'$ with two alternatives $w$ and $o$:
\begin{itemize}
    \item There is a group consisting of $n_{g_w}$ agents all of whom are located at the same point with distance $\gamma$ from $w$ and $y$ from $\OPT$.
    \item There are also $k-1$ groups with sizes equal to the sizes of the remaining groups in the original instance $I$. 
    The $n-n_{g_w}$ agents in all those groups are located at the same point with distance $z_w = C-x$ from $w$ and $z_\OPT=C+x$ from $\OPT$. 
\end{itemize}  
Observe that $w$ still minimizes the total distance of all agents in $I'$; this follows since the total distance of the agents from $w$ and $o$ remain the same as in $I$. For the same reason, $\cost(w\,|\,I') = \cost(w\,|\,I)$. 

We now argue that $\cost(\OPT\,|\,I') \leq \cost(\OPT\,|\,I)$. 
By the definition of $y$, the average distance of the agents in $g_w$ from $\OPT$ is the same as in $I$. 
Consider any group $g_o \neq g_w$ that maximizes $\frac{1}{n_{g_o}}\sum_{i \in g_o} d(i,\OPT)$. 
Since $g_o$ maximizes the average distance out of all groups that are different than $g_w$, we have that
\[
\frac{1}{n_{g_o}}\sum_{i \in g_o} d(i,\OPT) \geq \frac{1}{n-n_{g_w}}\sum_{g\not=g_w}\sum_{i \in g}d(i,\OPT) = z_\OPT = C + x.
\] 
Observe now that 
\[
\cost(\OPT\,|\,I) = \max \bigg(\frac{1}{n_{g_o}}\sum_{i \in g_o} d(i,\OPT), \frac{1}{n_{g_w}}\sum_{i \in g_w}d(i,\OPT)\bigg) \geq \max \left\{C+x, y\right\} = \cost(\OPT\,|\,I'), 
\] 
thus proving our claim. 

Next, we transform $I'$ into a new instance $I''$ with the desired property $C+x = y$ such that the distortion does not decrease. 
We consider the following two cases.

\medskip
\noindent 
{\bf Case 1: $y \geq C+x$ in $I'$.} 
Let $C' = y - x \geq C$ and consider the following instance $I''$:
\begin{itemize}
    \item There is a group consisting of $n_{g_w}$ agents all of whom are located at the same point with distance $\gamma$ from $w$ and $y$ from $\OPT$.
    \item There are also $k-1$ groups with sizes equal to the sizes of the remaining groups in the original instance $I$. 
    The $n-n_{g_w}$ agents in all those groups are located at the same point with distance $C'-x$ from $w$ and $C'+x$ from $\OPT$. 
\end{itemize}  
In this new instance $w$ still minimizes the total distance from all agents; indeed, the location of any agent $i \in g_w$ is the same in both instances $I'$ and $I''$, while any agent $i \not\in g_w$ has been moved closer to $w$ and $o$ by exactly the same distance $C'-C = y - (x+C) \geq 0$ between the two instances. Observe that $\cost(\OPT\,|\,I'') = \cost(\OPT)$ since the average distance of any group from $\OPT$ is $y$ in $I''$ and $\cost(\OPT\,|\,I') = \max\{C+x,y\} = y$ by our assumption for this case. In addition, the distances to $w$ increase for some agents, and thus $\cost(w\,|\,I'') \geq \cost(w\,|\,I')$, which further means that the distortion does not decrease as we go from $I'$ to $I''$ for which the desired property $y = C'+x$ holds.

\medskip
\noindent
{\bf Case 2: $y < C+x$ in $I'$.}
Let $y' = C+x$, $\gamma' = \gamma + (y' - y)$ and consider the following instance $I''$:
\begin{itemize}
    \item There is a group consisting of $n_{g_w}$ agents all of whom are located at the same point with distance $\gamma'$ from $w$ and $y'$ from $\OPT$.
    \item There are also $k-1$ groups with sizes equal to the sizes of the remaining groups in the original instance $I$. 
    The $n-n_{g_w}$ agents in all those groups are located at the same point with distance $C-x$ from $w$ and $C+x$ from $\OPT$. 
\end{itemize} 
As in the previous case, $w$ still minimizes the total distance from all agents; indeed, the any agent $i \in g_w$ has been moved closer to $w$ and $o$ by the same distance $y'-y$ between the two instances, while any agent $i \not\in g_w$ is at the same location in both instances. 
We also have that $\cost(\OPT\,|\,I'')=\cost(\OPT\,|\,I') = C+x$, and $\cost(w\,|\,I'') \geq \cost(w\,|\,I')$ since the distance to $w$ has increased for some agents from $y < C+x$ to $y'=C+x$. Hence, the distortion again does not decrease, and the new instance again satisfies the desired property that $y' = C+x$.
\end{proof}
    
By the above lemma, there exists a worst-case instance (where the ratio $\frac{\cost(w)}{\cost(o)}$ is maximized) with $C+x = y$. 
We consider the following two cases. 

\medskip
\noindent 
{\bf Case 1:  $x\geq \frac{n_{g_w}}{n-n_{g_w}}$}.
Then, since $\cost(w) = \gamma \leq 2+y = 2 + C + x$ and $\cost(\OPT) = \max\{y, C+x\} = C+x$, the distortion is at most 
\begin{align*}
    \frac{\cost(w)}{\cost(o)} = \frac{2 + C + x}{C+x}.
\end{align*}
This expression is a non-increasing function in terms of $C$ and $x$. Hence, since $C \geq 1$ and $x\geq \frac{n_{g_w}}{n-n_{g_w}}$, we obtain an upper bound of 
\[\frac{3 + \frac{n_{g_w}}{n-n_{g_w}}}{1+\frac{n_{g_w}}{n-n_{g_w}}} = 3 - \frac{2n_{g_w}}{n} \leq 3 - \frac{2\mu}{n}.\]

\medskip
\noindent 
{\bf Case 2: $x \leq \frac{n_{g_w}}{n-n_{g_w}}$.}
Since $z_w = C-x$, have that 
\[\sum_i d(i,w) = n_{g_w} \cdot \gamma + (n-n_{g_w})\cdot z_w = n_{g_w} \cdot \gamma + (n-n_{g_w})\cdot (C-x).\]
Also, since $z_o = C+x = y$, we have that
\[\sum_i d(i,o) = n_{g_w} \cdot y + (n-n_{g_w})\cdot z_o = n\cdot(C+x).\]
Using these and the definition of $w$, which is the alternative that minimizes the total distance from all agents, we further have that
\begin{align*}
\sum_i d(i,w) \leq \sum_i d(i,o) \Leftrightarrow \gamma &\leq \frac{2nx}{n_{g_w}} + C - x.
\end{align*}
Hence, the distortion in the worse case instance is at most,
\begin{align*}
\frac{\cost(w)}{\cost(\OPT)} = \frac{\gamma}{\cost(\OPT)} \leq \frac{\frac{2nx}{n_{g_w}} + C - x}{C+x}
\end{align*} 
This expression is a non-decreasing function of $x$ and a non-increasing function of $C$. Since $x \leq \frac{n_{g_w}}{n-n_{g_w}}$ and $C\geq 1$, we obtain an upper bound of 
\[\frac{\frac{2n}{n-n_{g_w}}+1-\frac{n_{g_w}}{n-n_{g_w}}}{1+\frac{n_{g_w}}{n-n_{g_w}}} = 3-\frac{2n_{g_w}}{n} \leq 3-\frac{2\mu}{n}.\]
The proof is now complete.
\end{proof}

\end{document}